\documentclass[%
 reprint,
 amsmath,amssymb,
 aps,
]{revtex4-2}
\usepackage{xcolor}
\usepackage{amsmath}
\usepackage{amsthm}
\newtheorem{claim}{Claim}[section]
\usepackage{graphicx}% Include figure files
\usepackage{dcolumn}% Align table columns on decimal point
\usepackage{bm}% bold math
\usepackage{comment}
\usepackage{placeins}

\begin{document}
\setlength{\textfloatsep}{10pt plus 1.0pt minus 2.0pt}
\preprint{APS/123-QED}

\title{Solving a directed percolation inverse problem}

\author{Sean Deyo}
 \email{sjd257@cornell.edu}
\affiliation{Cornell University\\
    Ithaca, NY}

\date{\today}

\begin{abstract}
We present a directed percolation inverse problem for diode networks: Given information about which pairs of nodes allow current to percolate from one to the other, can one find a configuration of diodes consistent with the observed currents? We implement a divide-and-concur iterative projection method for solving the problem and demonstrate the supremacy of our method over an exhaustive approach for nontrivial instances of the problem. We find that the problem is most difficult when some but not all of the percolation data are hidden, and that the most difficult networks to reconstruct generally are those for which the currents are most sensitive to the addition or removal of a single diode.
\end{abstract}


\maketitle

%\tableofcontents

\section{Introduction}\label{sec:intro}
%Percolation problems have long been a topic of interest in a variety of disciplines, including the chemistry of polymers \cite{flory1941polymer,stockmayer1944molecular}, the ecology of habitat fragmentation \cite{boswell1998habitat}, and the dynamics of commuter traffic networks \cite{li2015traffic}. The archetypal percolation problem is the flow of liquid through a porous medium. One can model the medium as a network in which the interstices are nodes and two interstices share an edge if they are near one another. Each edge can be ``open'' or ``closed'' depending on whether liquid can flow from one node to the other. The task is to ascertain whether the liquid has a path from one side of the medium to the other \cite{broadbent1957percolation}.

%If the edges are opened randomly with probability $p$, one can argue that in an infinite medium there is a critical $p$ below which there is never a path and above which there is always a path. A significant amount of research has focused on analytical \cite{bollobas2006thresholds} and numerical \cite{newman200monte} techniques for evaluating the critical $p$ and exploring the behavior of networks near the critical point \cite{stauffer2018introduction}.

%One can also introduce a bias that skews current in one direction more than the other \cite{blease1977series}. Many authors take one dimension of a multidimensional lattice to represent time and only allow flow forward in time \cite{grassberger1989directed}. This makes it possible to map the problem onto the Ising model or cellular automata \cite{domany1984equivalence}.

{ Directed percolation (DP) is a model of spreading processes with a directional bias, such as the flow of fluid downward through a porous medium, the spread of a forest fire under the influence of wind, or the growth of a polymer in a flowing liquid~\cite{obukhov1980problem}. It can also apply to more abstract processes like phase transitions in liquid crystals~\cite{takeuchi2007liquid}, interface pinning~\cite{tang1992pinning}, and the proliferation of turbulence in hydrodynamics~\cite{pomeau1986hydrodynamics}. 

%DP becomes even more powerful if one interprets the flow direction as the time dimension, as it shares universal critical behavior with Reggeon field theory~\cite{cardy1980reggeon,grassberger1989directed}. 
If one interprets the flow direction as the time dimension, DP becomes an elegant formulation for the time evolution of the Ising model or a cellular automaton~\cite{domany1984equivalence}. This interpretation comes with a natural inverse problem: If one observes the state of a cellular automaton at two different times, one can attempt to reconstruct the rules that govern the system's evolution~\cite{springer2021gameoflife,elser2021reconstructing}. Though this interpretation involves a regular and repeating network such as a square lattice, it is possible to study percolation in random~\cite{fan2012random,lee2018complex} or clustered networks~\cite{colomer2014double,miller2009cluster} that are more apt to model social networks or the spread of disease. It is also possible to make percolation locally directed without necessarily favoring one particular direction globally, by randomly installing diode-like edges that allow flow one way but not the other~\cite{broadbent1957percolation}. One can then calculate properties such as the conductivity or percolation threshold of lattices with such edges~\cite{redner1982directed,redner1982conductivity,verbavatz2021oneway,denoronha2018isotropic}.

%The problem of reconstructing such networks has received some attention in recent years~\cite{peixoto2019community,angulo2017fundamental}. Like the cellular automaton problem, the reconstruction of these networks is often based on data in the form of a time series~\cite{runge2018causal}.

%The focus of this paper is closer to the cellular automaton interpretation than to the classic statistical study of DP. Indeed, the ``divide-and-concur'' method, which we employ in our algorithm, has been used to reconstruct the rules of a cellular automaton based on time series data~\cite{elser2021reconstructing}. 
}

\begin{figure}[t]
    \centering
    \includegraphics[width=.45\textwidth]{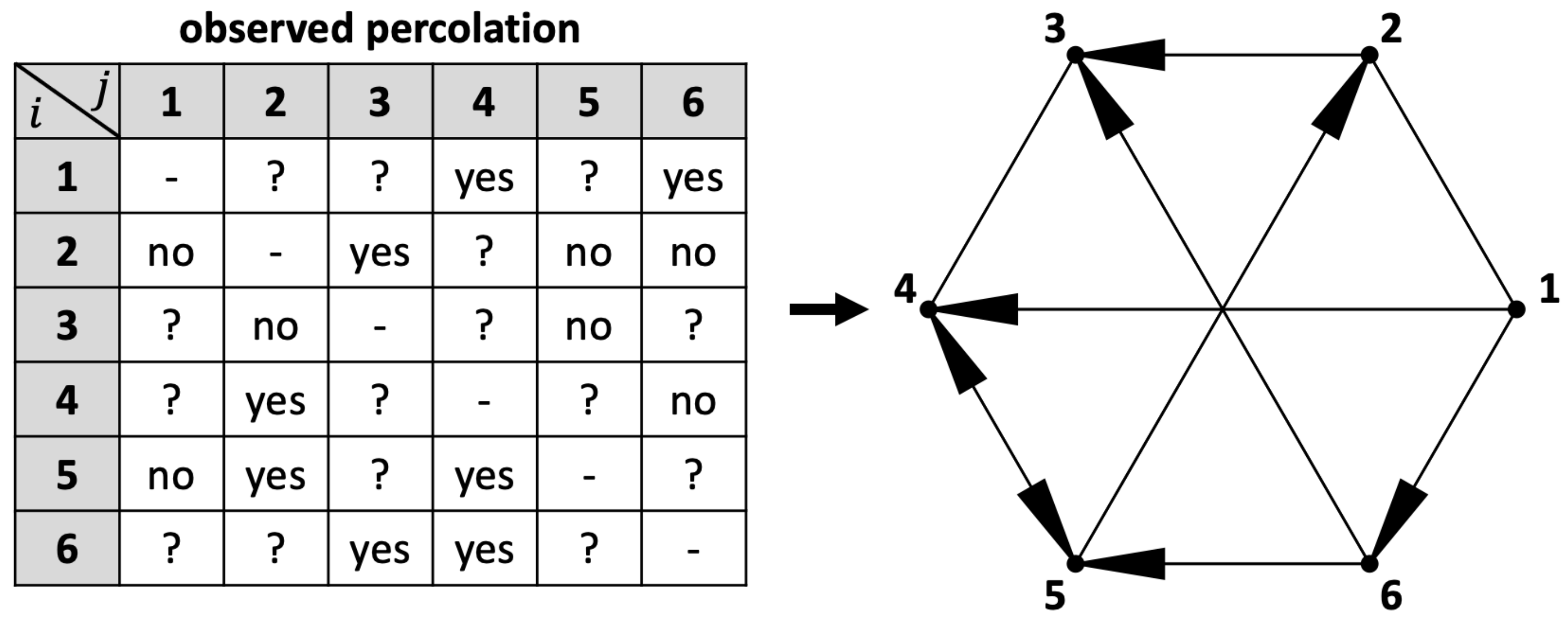}
    \caption{The inverse percolation problem. Given data (left) about whether current can percolate from $i$ to $j$, the task is to find a network (right) consistent with the observations. One does not necessarily have current data for every $(i,j)$; in fact, the problem is nontrivial only if some of the data are hidden, which we denote here with question marks.}
    \label{fig:prob}
\end{figure}

The problem we study in this paper uses diodes as the source of directionality, but not on a lattice. Imagine an electrical circuit consisting solely of nodes, wires, and ideal diodes. Each wire can have
\begin{itemize}
    \item zero diodes, so that current can flow freely in either direction,
    \item one diode, so that current can flow in one direction but not the other, or
    \item two diodes back-to-back, so that current cannot flow in either direction.
\end{itemize}
Suppose we know which nodes share a wire but we do not know the location or number of diodes. We can apply a voltage across a pair of nodes $(i,j)$ and record whether current percolates from $i$ to $j$. We can do this for several (not necessarily all) pairs of nodes, as in Figure \ref{fig:prob}. Our task is to find a configuration of diodes consistent with these measurements. Many configurations could yield the same data, so we cannot ask for \textit{the} configuration that generated the data but merely \textit{a} configuration consistent with the data. Though we describe it in terms of diodes, this problem can also serve as a model for other systems in which flow is locally asymmetric but occurs in more than one direction globally, such as the spread of information in a social network or the movement of goods or people in a transportation network. 

{ We can pose the problem formally in terms of the adjacency matrix $M$: Let $M_{ij}=1$ if nodes $i$ and $j$ share an edge and that edge allows current from $i$ to $j$; otherwise, let $M_{ij}=0$. Then $\left[M^\ell\right]_{ij}$ gives the number of paths of length $\ell$ from $i$ to $j$. Asking if current percolates from $i$ to $j$ is equivalent to asking whether $\left[\exp(M)\right]_{ij}>0$. Thus, if we know that some elements of $\exp(M)$ are positive, some are zero, and others are completely unknown, our task is to find an $M$ consistent with these measurements.

Ours is an inverse problem, like the reconstruction of a cellular automaton, but without the spatial regularity of a lattice or the temporal ordering of time series data. The fact that we are inferring network information from percolation information, rather than the other way around, is perhaps our most salient departure from the canon of percolation problems. This fact makes our problem one of network reconstruction. Such problems are particularly relevant when observational data are available but performing a direct experiment is not possible for logistical or ethical reasons~\cite{runge2018causal}, as in the spread of an infectious disease. We diverge from most other network reconstruction methods in the discreteness and exactness of the solution we seek. Others might search for continuous parameters representing the strength of the interaction between two nodes~\cite{sontag2008steady} or the concentration of a substance at one node~\cite{angulo2017fundamental}, or they might seek to winnow down the space of possible networks probabilistically to identify the most likely candidates~\cite{peixoto2019community}. These methods are appropriate when the network is large, data are relatively plentiful, and probabilistic or continuous descriptions of the network are desirable. We instead seek solutions that definitively specify whether each edge is open and reproduce the observations exactly. These rigid constraints make our method most suitable for reconstructing small networks from scant or incomplete data.

It is possible to convert our problem into a Boolean satisfiability (SAT) problem. Appendix \ref{app:SAT} describes how to do so. The SAT problem was the first to be proven NP-complete~\cite{cook1971complexity}, and many other problems---sudoku~\cite{lynce2006sudoku}, cryptography~\cite{massacci2000cryptanalysis}, cellular automaton reconstruction~\cite{elser2021reconstructing}, the $n$-queens problem and others~\cite{bright2019effectiveSAT}---can be posed as SAT problems, making it a useful way to generically evaluate the difficulty of a logical constraint problem. There are several methods~\cite{molnar2018maxsat,ercsey2011sat,sorensson2005minisat,biere2013lingeling,gravel2008divide} available for solving SAT problems. Mapping our problem into SAT makes it possible in principle to tackle our problem with these methods; however, the number of variables and clauses involved in the corresponding SAT problem scales with the number of self-avoiding paths in the network, which can grow quite quickly with the number of nodes $n$ in the network (see appendix \ref{app:SAT}). In a complete network, for example, the number of self-avoiding walks is more than $n!$. In the method we present in Section \ref{sec:algo}, the number of variables scales with the number of edges times the number of cells in the data table, a product which is less than $n^4$.
}

\section{Algorithm}\label{sec:algo}
%To give a solution we must specify where there are diodes. The number of possible diodes corresponds to the number of directed edges in the network: If directed edge $i\to j$ has a diode, then current cannot flow from $i$ to $j$ along the wire that $i$ and $j$ share (though it could flow along some other path from $i$ to $j$). For every directed edge $e$ we create a diode variable $x_e$. If $x_e=1$, there is a diode blocking $e$. If $x_e=0$, there is no diode.

%The divide-and-concur method breaks the difficult task of finding a set of $x_e$ consistent with all observations into a collection of easy tasks: For each ordered pair of nodes $(i,j)$, find a set of $x_e$ consistent with cell $ij$ of the data table:
%\begin{itemize}
%    \item If current percolates from $i$ to $j$, there must be an unblocked path from $i$ to $j$; that is, there must be at least one path $w$ with $x_e=0$ for every edge $e$ along $w$.
%    \item If current does not percolate from $i$ to $j$, all paths from $i$ to $j$ must be blocked; that is, every path $w$ must have at least one $e$ such that $x_e=1$.
%    \item If there is no measurement for $(i,j)$, then every $x_e$ can be either $0$ or $1$.
%\end{itemize}
%There are $n(n-1)$ pairs to consider, where $n$ is the number of nodes, so this yields $n(n-1)$ copies of $x_e$ for each $e$. Let $x_{ep}$ denote the copy of $x_e$ for pair $p$, and let $x$ denote the vector containing all of the $x_{ep}$'s.

{ The divide-and-concur approach~\cite{gravel2008divide} divides a difficult problem with many constraints into a set of easy problems, finds a solution to each easy problem, then tries to make the solutions concur with each other. To apply this approach to our problem, we break the difficult task of finding a configuration of diodes consistent with all observations into a collection of much simpler tasks: For each ordered pair of nodes $p=(i,j)$, find a configuration of diodes consistent with cell $ij$ of the data table. There are $n(n-1)$ pairs to consider, where $n$ is the number of nodes, so this yields $n(n-1)$ configurations. If all of these configurations are the same, then we have solved the problem. If not, we must find a new set of configurations and iterate until they all concur.

Our algorithm formulates this task as a search for a point in the intersection of two sets in a high-dimensional Euclidean space. Let $E$ be the set of directed edges in the network. For each pair $p$ we create a set of variables $x_p=\{x_{ep}\,|\,e\in E\}$. A setting of $x_{ep}=1$ means there is a diode blocking directed edge $e$ in the configuration generated for pair $p$, and $x_{ep}=0$ means there is no such diode, but during the search process we will allow $x_{ep}$ to be any real number. We use the shorthand $x$ to denote the vector containing all of the $x_{ep}$'s.

Let $A$ be the set of $x$ such that all $x_{ep}$ are $0$ or $1$, and for each $p=(i,j)$ the diode configuration encoded by $x_p$ is consistent with cell $ij$ of the data table. The projection of $x$ to set $A$ is the point $x^A\in A$ such that the distance
\begin{equation}
    d(x,x^A)=\sum_{e,p} \left(x_{ep}-x^A_{ep}\right)^2
\end{equation}
is as small as possible. Projecting an arbitrary $x$ to the nearest point in $A$ can be computationally expensive, so we use a quasi-projection: It always gives a point in set $A$, but not necessarily the distance-minimizing point. 

The quasi-projection is implemented as follows. We know we must set all $x^A_{ep}$ to either $0$ or $1$. So our first step is to set
\begin{equation}
    x^A_{ep}=\begin{cases} 0 & x_{ep}<0.5\\
    1 & x_{ep}\ge0.5\\ \end{cases}
    \label{eq:Around}
\end{equation}
for all $e$ and $p$. If all we had to do was set the variables to $0$ or $1$, then we would be done here and this would be an exact projection. However, ensuring that each $x^A_p$ agrees with the data will involve changing some $0$'s to $1$'s and vice versa. The squared distance for setting $x^A_{ep}=0$ is $$\left(x_{ep}\right)^2$$ and for $x^A_{ep}=1$ it is $$\left(1-x_{ep}\right)^2=1-2x_{ep}+\left(x_{ep}\right)^2,$$ so $1-2x_{ep}$ is the ``extra distance'' for changing $x^A_{ep}$ from $0$ to $1$, and $2x_{ep}-1$ is the extra distance for changing a $1$ to a $0$.

Now, to ensure $x^A_p$ agrees with the data for $p=(i,j)$:
\begin{itemize} 
    \item If there are no data for $p$, do nothing.
    \item If the data indicate current percolates from $i$ to $j$, check if there is at least one path from $i$ to $j$ with $x^A_{ep}=0$ along every directed edge $e$ of the path. If there is no such path, rank the paths by the extra distance for opening the path: For a given edge $e$, if $x_{ep}\geq0.5$ then we have $x^A_{ep}=1$ and we incur an extra distance of $2x_{ep}-1$ to change $x^A_{ep}$ to $0$. If $x_{ep}<0.5$ then we already have $x^A_{ep}=0$ and there is no extra distance. Thus, we can write the total extra distance for the path as
    \begin{equation*}
        \sum_{e\in\text{path}}
    \min(2x_{ep}-1,0).
    \end{equation*} For the path with the lowest sum, set $x^A_{ep}=0$ for every $e$ on the path.
    \item If the data indicate current does not percolate from $i$ to $j$, then for every path from $i$ to $j$, check if at least one edge $e$ along the path has $x^A_{ep}=1$. For any path for which this is not so, choose the $e$ on that path with the largest $x_{ep}$ (i.e., smallest $1-2x_{ep}$) and set $x^A_{ep}=1$.
\end{itemize}
This last point, the method for blocking current, is what makes our projection a quasi-projection, as it is not strictly distance-minimizing. See Appendix \ref{app:projA} for an illustrative example.
}

%We formulate our algorithm in terms of two constraint sets. Set $A$ consists of all $x$ such that for each $p=(i,j)$, the set of $x_{ep}$ is consistent with cell $ij$ of the data table. Projecting an arbitrary $x$ to the nearest point in $A$ can be computationally expensive. We use a quasi-projection, detailed in Appendix \ref{app:projA}, which ``projects'' $x$ to a point that is in $A$ but is not necessarily the closest such point. Let $x^A$ denote this quasi-projection of $x$ to $A$.

Projecting to $A$ solves all of the easy problems; what remains is to make these solutions concur. For this we define $B$ as the set of $x$ such that for each $e$, all $x_{ep}$ are equal. (This does \textit{not} require that $x_{ep}=0$ or $1$.) Projecting to $B$ is a simple average:
    \begin{equation}
        x^B_{ep}=\frac{1}{n(n-1)}\sum_{p'} x_{ep'}.
    \end{equation}
Any $x\in A\cap B$ is a solution to our problem. 

We initialize $x$ with random real numbers between $0$ and $1$, then iterate using a generalized Douglas-Rachford method~\cite{aragon2020douglas}:
\begin{equation}
    x \to x + \beta\, \frac{R_B(R_A(x))-x}{2}
    \label{eq:iter}
\end{equation} where $R_A(x)=2x^A-x$ is the reflection of $x$ across $x^A$, and likewise for $R_B$. { Using reflections, rather than just projections, is critical to allow the algorithm to escape from non-solution traps where $A$ and $B$ come close but do not intersect~\cite{aragon2020douglas}.} One can check that if $x$ is a fixed point of \eqref{eq:iter}, then $x^A \in A\cap B$. The parameter $\beta$ controls the size of the iteration steps. Taking ${\beta\to0}$ gives a reliable, continuous trajectory but makes the search process slow. Larger $\beta$ ($\sim 1$) can make the search quicker but sometimes leads the algorithm to get stuck in limit cycles. In this paper we use a conservative value, $\beta=0.1$, so that the algorithm avoids this limit-cycle trapping in all but a few percent of trials.

{Finally, we note that there are other realizations of the divide-and-concur principle for this problem. Having briefly tested some alternatives, we can say that the choice of implementation can affect the algorithm's performance but for the most part does not affect the major patterns we will describe in Section \ref{sec:exp}.
\begin{itemize}
    \item One speed-up is to have $x_p$ only for the pairs $p$ for which the percolation data are known. If some of the data are hidden, this mean fewer pairs to handle in the $A$ projection, which makes each iteration faster. However, we find that it can also make the algorithm more likely to get trapped in limit cycles, particularly when only two or three pairs in the data table are known.
    \item One can also take advantage of the percolation data for ordered pairs that happen to share an edge: For every directed edge $e=i\to j$, if we observe no percolation from $i$ to $j$, we can permanently set $x_{ep}=1$ for all $p$. If we do observe percolation from $i$ to $j$, we can safely set $x_{ep}=0$ for all $p$. By not using these shortcuts in our implementation we are effectively asking the algorithm to be clever enough to discover them on its own.
    \item If speed is no concern but the exactness of the projection is, one can create variables $x_{ew}$ for every path $w$, not just every pair. This makes our approach of blocking path by path an exact projection, but the enormous number of variables makes the algorithm much slower.
\end{itemize}
This is not an exhaustive list of alternative implementations, but it illustrates the diversity of approaches to solving a problem in the divide-and-concur framework. The framework is also quite adaptable to extra global constraints. For example, if one wanted to impose a maximum or minimum number of diodes, one could add an extra step in the $B$ projection that computes $\sum_e x_{ep}^B$ and, if the sum does not satisfy the max/min requirement, add a constant to all of the $x_{ep}^B$ to ensure it does. Implementing such a constraint would be less straightforward if one chose to solve this problem with, say, the SAT formulation in Appendix~\ref{app:SAT}.
}

\section{Experiments}\label{sec:exp}
In the experiments that follow, we avoid trivial instances of the problem as much as possible. One obvious trivial instance is when all of the current data are hidden, in which case there are no constraints and thus any configuration of any number of diodes is a solution. However, when none of the data are hidden there are more than enough constraints---enough for one to construct a solution by hand{: For any $(i,j)$ that share a wire and for which current does not flow from $i\to j$, place a diode on wire $ij$ to prevent current from $i\to j$. These diodes alone are sufficient to solve the problem. (See Appendix \ref{app:trivialproof} for a proof.)} So to keep things nontrivial we hide exactly half of the data, except in section \ref{subsec:hidden} when we explore how the difficulty of the problem depends on the number of hidden data. 

The problem is also trivial if there are no wires, in which case there is no current and nowhere to put any diodes. But in a network with a wire between every pair of nodes, once again there is a by-hand construction{: Place diodes blocking every directed edge $i\to j$ except those for which we know current can flow from $i$ to $j$. Even if the data are incomplete, these diodes solve the problem. (See Appendix \ref{app:trivialproof} for a proof.)} So apart from section \ref{subsec:wires}, in which we vary the number of wires, we use a type of bipartite network that has roughly half of its possible wires: We take an even number of nodes and let nodes $i$ and $j$ share a wire if $i+j$ is odd.

\begin{figure}[t]
    \centering
    \includegraphics[width=.48\textwidth]{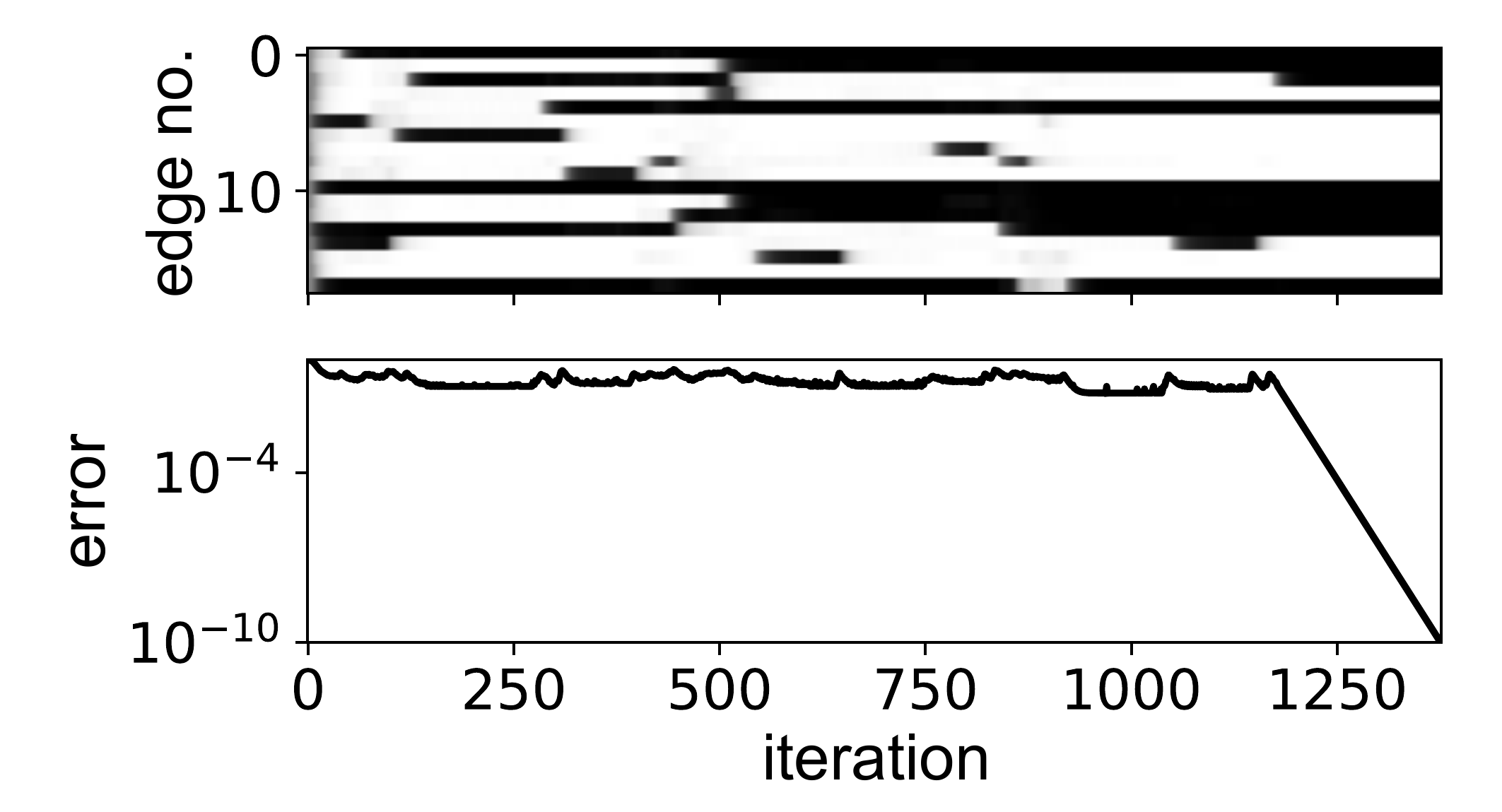}
    \caption{Top: Evolution of the concur estimate $x^B$ as the algorithm searches for a solution to a network with $n=6$ nodes. Each row represents a single directed edge $e$, with lighter color corresponding to larger $x^B_e$. Bottom: The evolution of the error (defined in the main text) for the same trial. The precipitous drop in error after about $1200$ iterations indicates that the algorithm has found a solution.}
    \label{fig:solving}
\end{figure}

The example in Figure \ref{fig:prob} is a typical nontrivial instance---a bipartite network with half of the data hidden. Figure \ref{fig:solving} illustrates the behavior of our algorithm as it tries to solve such an instance. The upper panel shows the evolution of the concur estimate $x^B$, with each row corresponding to $x^B_e$ for a single directed edge $e$. (Since $x^B_{ep}$ is the same for all $p$, we can drop the $p$ from the subscript.) White corresponds to $x^B_e$ closer to $1$, meaning that the algorithm thinks there should be a diode on edge $e$, while black indicates no diode. Gray shades indicate the algorithm is uncertain, or is ``changing its mind'' and adding or removing that diode. The lower panel plots the error, which we define as the RMS of $$\frac{R_B(R_A(x))-x}{2}\,.$$ Comparing the upper and lower panels of Figure \ref{fig:solving}, one can see that the spikes in the error time series coincide with moments when the algorithm adds or removes a diode. Suddenly, after about $1200$ iterations the error drops by many orders of magnitude, reflecting the ``aha moment'' when the algorithm has found a solution. In practice, we let the algorithm stop and declare success once the error falls below $10^{-3}$.

\subsection{Scaling with $n$}\label{subsec:scaling}
\setlength{\tabcolsep}{6pt}
\begin{table}[t]
\centering
    \begin{tabular}{r|l|l}
        $n$ & projection & exhaustive \\
        \hline
        $4$ & $0.00040(4)$ & $0.000018(1)$ \\
        $6$ & $0.008(1)$ & $0.058(6)$ \\
        $8$ & $0.15(2)$ & $3.0(2)\times10^3$ \\
        $10$ & $8(1)$ &  \\
    \end{tabular}
    \caption{Average time, in seconds, to solve an instance of the problem for several values of $n$, comparing our projection method to an exhaustive ``brute force'' approach.}
    \label{tab:scaling}
\end{table}

The most obvious parameter of our problem is the number of nodes, and it is important to know how our method scales with $n$. For each even $n$ from $4$ to $10$ we randomly generated $100$ networks, with the number of diodes chosen to maximize the difficulty of the problem for our algorithm (see Section \ref{subsec:diodes}), then hid half of the current data. To provide a comparison for our algorithm, we also implemented an exhaustive ``brute force'' solution finder, which simply lists all of the possible solutions and checks them one by one until finding one that works. 

Table \ref{tab:scaling} gives the average time required for our projection method and the exhaustive approach. The computational expense grows quite quickly for the projection algorithm: From $n=4$ to $n=10$ the time required increases by five orders of magnitude. However, the exhaustive approach scales much worse: The required time grows by eight orders of magnitude just from $n=4$ to $n=8$. We have not attempted using the exhaustive checker for $n=10$.

One can understand the growth of the search time by considering what each method actually does. In principle, each iteration of our algorithm may involve checking every possible path for current from any one node to any other. We only consider self-avoiding paths, because any path that visits a node more than once is irrelevant as far as current is concerned, but the number of such paths can still grow with $n!$. Meanwhile, the exhaustive approach has to contend with $n^2/2$ undetermined diodes, yielding $2^{n^2/2}$ possible solutions. Factorial growth may be fast, but $2^{n^2/2}$ grows faster. For large networks the exhaustive approach is simply infeasible.

\subsection{Hiding data}\label{subsec:hidden}
In Appendix \ref{app:trivialproof} we prove that the problem is trivial when all or none of the data are hidden. We will now explore how the algorithm behaves in the two trivial extremes and all the cases in between.

\begin{figure}[t]
    \centering
    \includegraphics[width=.45\textwidth]{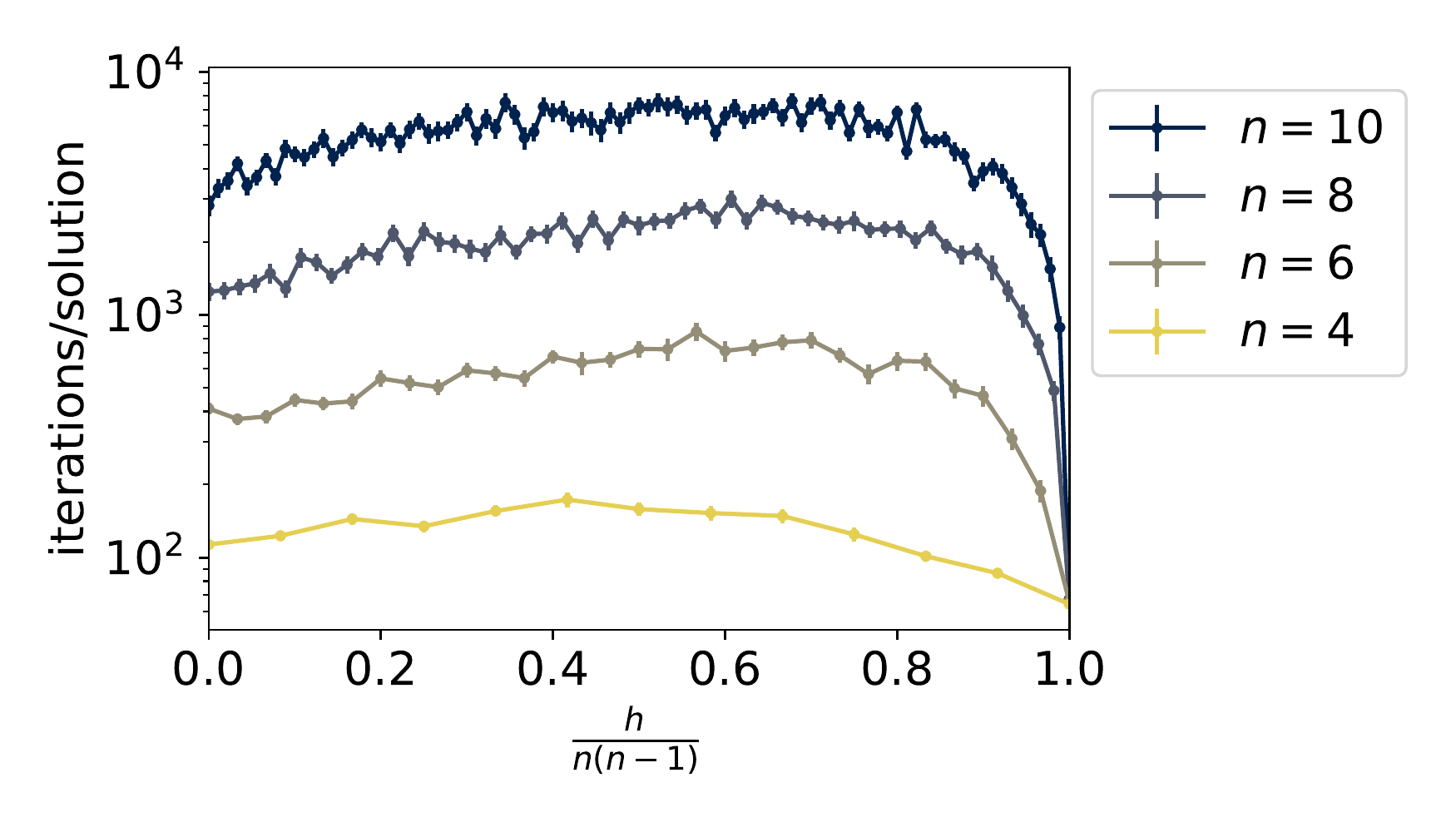}
    \caption{Average number of iterations required to find a solution as a function of the number of hidden data $h$ for several values of $n$.}
    \label{fig:ht}
\end{figure}

For a given $n$ the number of hidden data $h$ can be any integer from $0$ to $n(n-1)$. For each $h$ we generated $100$ instances, once again with the number of diodes for each $n$ chosen to maximize the difficulty, and recorded the average number of iterations required to find a solution. 

Figure \ref{fig:ht} plots the results. Beginning from $h=0$ the required number of iterations increases as more data are hidden. The logarithmic scale disguises the magnitude of this effect somewhat: The number of iterations when half of the data are hidden is actually about twice as many as when none of the data are hidden. However, as $h$ continues increasing there comes a point at which the algorithm begins to find solutions much more quickly, particularly when $h$ reaches $n(n-1)$.

Every entry in the data table is a constraint that limits the space of possible solutions. When very few of the data are hidden the algorithm seems to take advantage of the many constraints, finding solutions in fewer and fewer steps as $h$ approaches $0$. A human does much the same thing in constructing a solution by hand for $h=0$ (see Appendix \ref{app:trivialproof}). On the other hand, when most of the data are hidden, the problem is underconstrained and there are many possible solutions. It is natural that the algorithm finds solutions quite easily in this regime---especially when $h=n(n-1)$, as that is the case in which any network is a solution.

\subsection{Number of diodes}\label{subsec:diodes}
Next we explore how the number of diodes affects the difficulty of the problem. The number of diodes $d$ can be any integer from $0$ to the number of directed edges $|E|$, which is $n^2/2$ for our bipartite networks. For each $d$ we have the algorithm attempt to solve $100$ networks with $d$ randomly placed diodes. In each instance we hide exactly half of the data and record the average number of iterations required to reach a solution.

\begin{figure}[t]
    \centering
    \includegraphics[width=.45\textwidth]{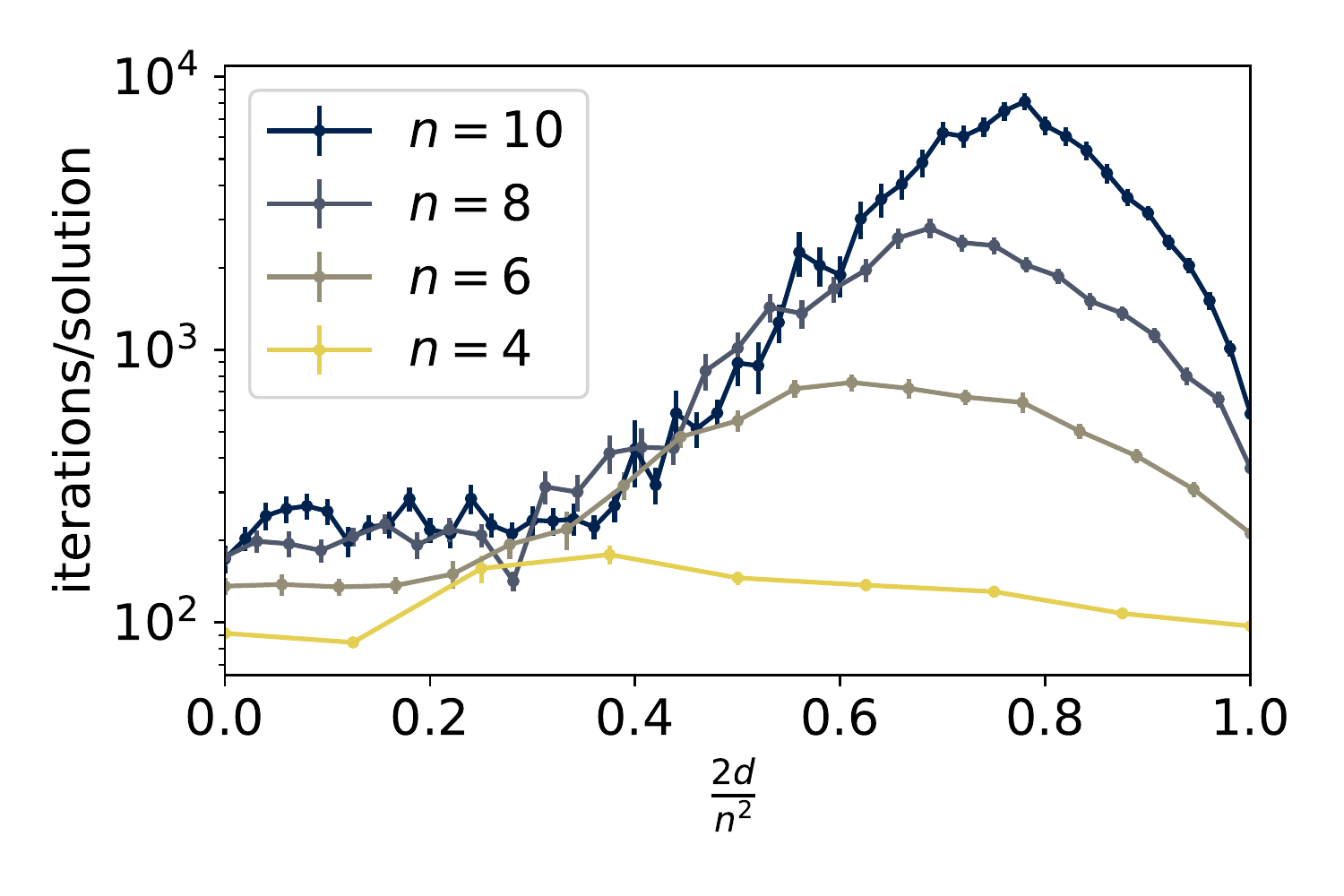}
    \caption{Average number of iterations required to find a solution as a function of the number of diodes $d$ for several values of $n$.}
    \label{fig:diter}
\end{figure}

Figure \ref{fig:diter} plots the results. For each $n$ the required number of iterations is generally small when $d$ is near $0$ or $n^2/2$, with a peak somewhere in between. The peak shifts farther to the right as $n$ increases.

\begin{figure}[t]
    \centering
    \includegraphics[width=.45\textwidth]{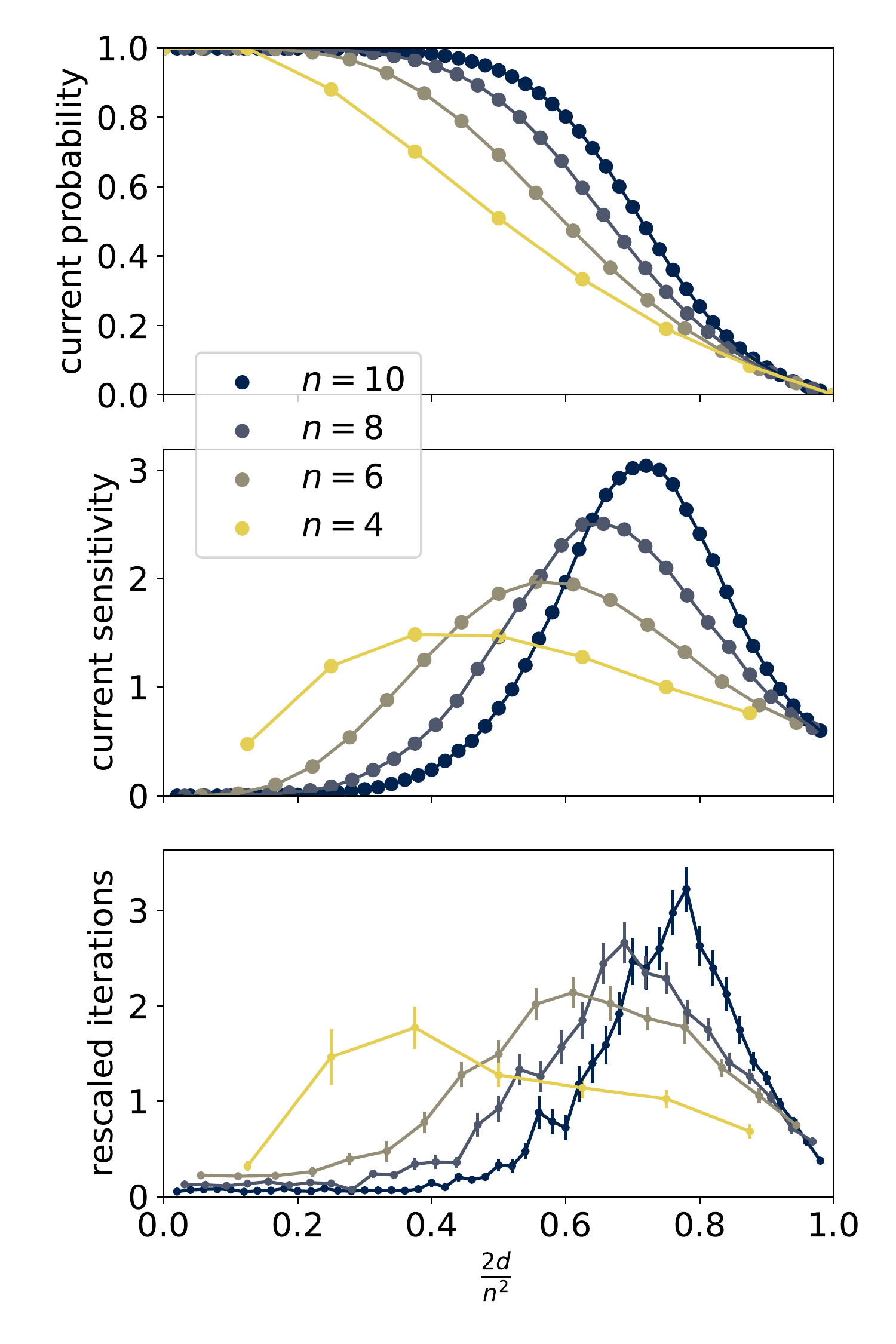}
    \caption{Top: Probability that current can percolate from one randomly chosen node to another, as a function of the number of diodes $d$ in the network. Middle: The current sensitivity, which is the magnitude of the slope of the upper panel. Bottom: The same iterations as in Figure \ref{fig:diter}, after rescaling as described in the text.}
    \label{fig:dI}
\end{figure}

In order to understand this behavior, consider the influence of the diodes on the current percolations of a network. The top panel of Figure \ref{fig:dI} shows the current probability---the likelihood that current can percolate from one randomly chosen node to another---as a function of $d$. There is a range of small $d$ for which current almost always flows, and this range extends farther to the right as $n$ increases. As $d\to n^2/2$, the current probability for any $n$ approaches $0$ linearly.

Both limits make sense intuitively: Current is very likely to flow when the number of diodes is small. Having more nodes allows many more paths, so as $n$ increases it takes many more diodes to block all the paths from one node to another. On the other hand, when $d$ is almost maximal the number of node-pairs for which current percolates is simply the number of directed edges that have not been blocked by a diode, $n^2/2-d$, which does indeed approach $0$ linearly as $d\to n^2/2$.

The middle panel of Figure \ref{fig:dI} gives what we call the current sensitivity, which is simply the magnitude of the slope of the upper panel. The slope is essentially the typical number of percolations that are affected by adding or removing a single diode, so a steeper slope corresponds to a more sensitive network. The current sensitivity curves are qualitatively similar to the curves in Figure \ref{fig:diter}: small on either side, with a peak that shifts to the right as $n$ increases. 

To make the similarity more obvious, in the lower panel of Figure \ref{fig:dI} we re-plot the data from Figure \ref{fig:diter} with a few modifications: First, we plot with a linear rather than logarithmic vertical axis. Second, we subtract off the number of iterations it takes for a completely trivial instance with all of the data hidden. We are essentially removing the part of the search at the far right of Figure \ref{fig:solving}, when the algorithm has found an answer and is simply converging to reach our desired stopping error ($10^{-3}$). This turns out to be about $65$ iterations, regardless of $n$. Finally, since the larger networks take significantly more iterations, we divide the number of iterations by the square of the number of possible diodes, $n^4/4$. Our original motivation for doing so is simply that it makes all of the curves roughly the same magnitude and therefore easy to plot on a linear scale, but in fact it turns out to be a very good rescaling---the plots of this modified iteration number are remarkably similar to the plots of current sensitivity, not only in shape but also in magnitude.

We are hesitant to make any strong claims about whether the similarity of the magnitudes is meaningful, but we are confident that the similar shapes are no coincidence. The networks that require more iterations to solve tend to be those with greater current sensitivity. In other words, the hardest networks to solve are those for which the current percolations are most affected by adding or removing a single diode. Since the algorithm's task is to match the percolation data by adding and removing diodes, it makes sense that the task is more difficult when the currents are more sensitive.

\subsection{Number of wires}\label{subsec:wires}
Finally, we depart from our bipartite networks and explore how the number and placement of wires affects the difficulty of the problem. For a network with $n$ nodes the number of wires $w$ can be any integer from $0$ to ${n(n-1)/2}$. Here we will give results for networks of size $n=6$. The behavior for other values of $n$ is similar. For each $w$, the number of diodes $d$ can be any integer from $0$ to $2w$. For each $d$ we generated $1000$ networks, hid half of the data, and recorded the number of iterations needed to find a solution.

There are many ways in which one can remove wires. One option is to prune at random. Another approach is to disconnect the network as quickly as possible: Pick a node and remove all the wires that connect to it before touching any other wires. Let us call this the recursive method, because once the chosen node is disconnected one can repeat the process with the remaining $n-1$ nodes. Alternatively, one might delay disconnection as long as possible. One way to do this is to pick two nodes and remove the wire they share. Then pick a third node and remove the wires it shares with the first two. Then pick a fourth node and remove the wires it shares with the first three, and so on. We call this the star method, because once we have picked all but one of the nodes we are left with a star graph (in which the unpicked node is the center and the $n-1$ remaining wires connect the center to the other nodes). Removing the last wires from the star completes the pruning.

\begin{figure}[t]
    \centering
    \includegraphics[width=.45\textwidth]{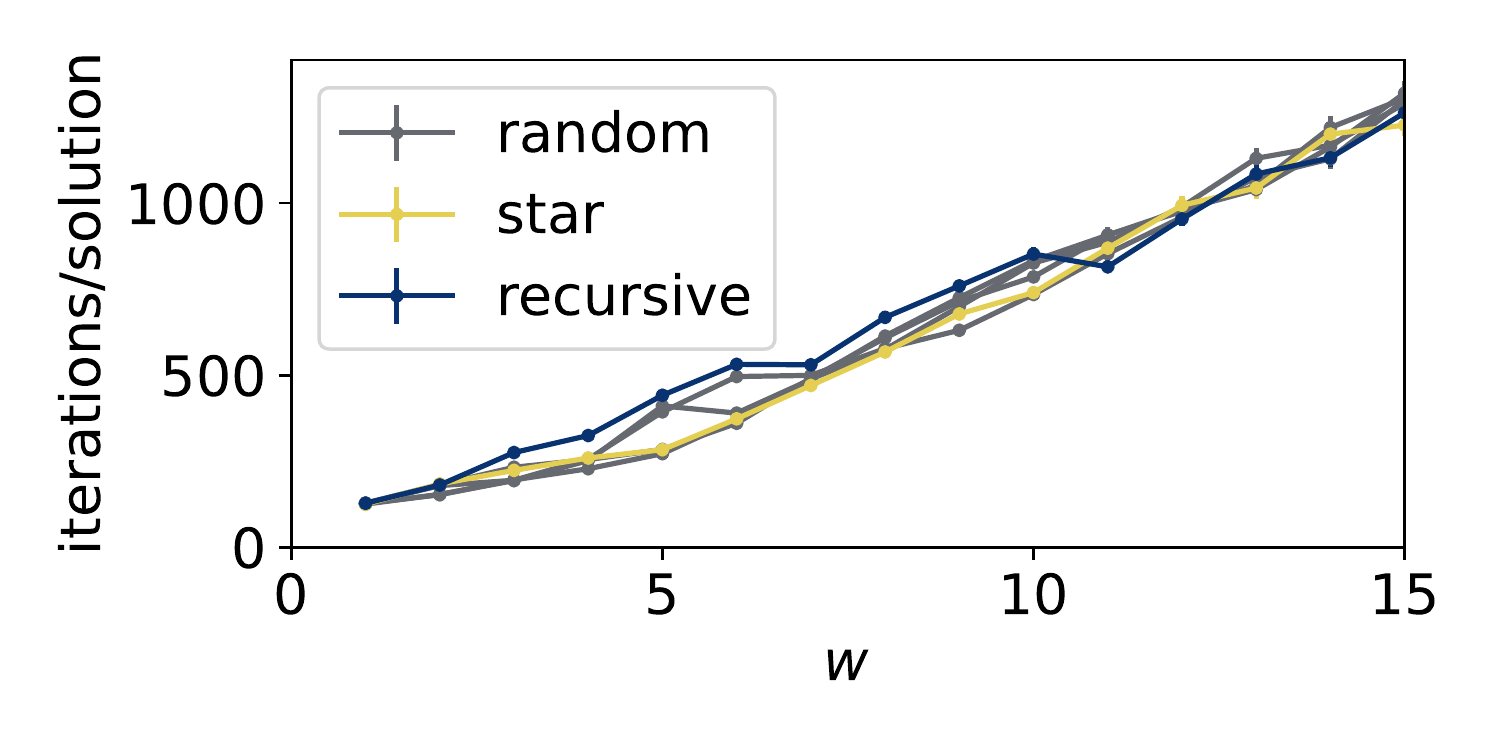}
    \caption{Average number of iterations as a function of the number of wires $w$ for networks with $n=6$ nodes. The dark curve is for the recursive pruning method; the light curve, the star method; the gray curves are several instances of random pruning.}
    \label{fig:n6witer}
\end{figure}

How much does the pruning order matter? To find out, in Figure \ref{fig:n6witer} we plot the average iteration count as a function of $w$ for each of the three methods. As in sections \ref{subsec:scaling} and \ref{subsec:hidden}, we are choosing $d$ to maximize difficulty: For each $w$ we select the $d$ that requires the most iterations to solve. The difficulty almost always grows monotonically with $w$, with little difference between the methods. The disparity between the recursive and star methods is greatest at $w=5$, the moment at which the star method has produced a star graph. An astute observe may notice that the iteration count for the recursive method is higher at $w=10$ than $w=11$. This coincides with the moment the network of $n=6$ nodes becomes disconnected, leaving a single isolated node and the other $5$ nodes completely connected by $10$ wires. The same phenomenon occurs between $w=7$ and $w=6$, the moment when the next node becomes isolated and the remaining $4$ are completely connected.

If all one cares about is the most difficult diode configuration for a given $w$, then the order in which one removes wires does not seem to matter much. But what about the pattern from section \ref{subsec:diodes} in which the iteration count was proportional to the current sensitivity: Is this true for all $w$? 

\begin{figure}[t]
    \centering
    \includegraphics[width=.45\textwidth]{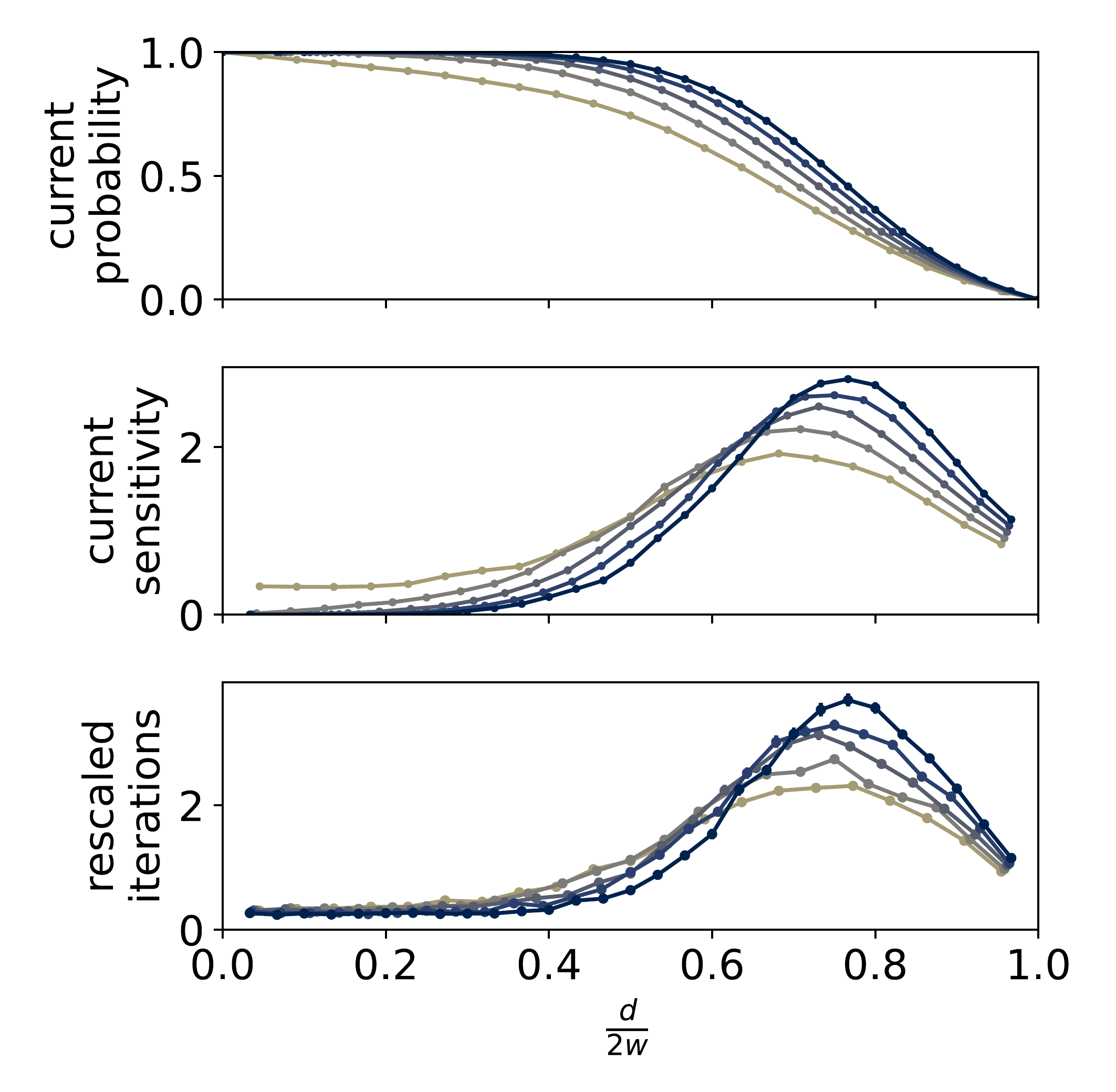}
    \caption{Current probability, current sensitivity, and rescaled iteration count for the recursive pruning method. The colors indicate the number of wires, with the darker curves corresponding to larger $w$.}
    \label{fig:wdIrecursive}
\end{figure}

Figure \ref{fig:wdIrecursive} plots the current probability, current sensitivity, and average number of iterations (rescaled in the same way as in section \ref{subsec:diodes}) for the recursive pruning method, with lighter color corresponding to fewer wires. Thanks to recursion we only need to plot the results from $w=15$ down to $w=11$. After that it reduces to the analogous problem for $n=5$. These results echo those of section \ref{subsec:diodes} in that the rescaled iteration count is roughly proportional to the current sensitivity, with the possible exception of the lightest curve, corresponding to $w=11$.

\begin{figure}[t]
    \centering
    \includegraphics[width=.45\textwidth]{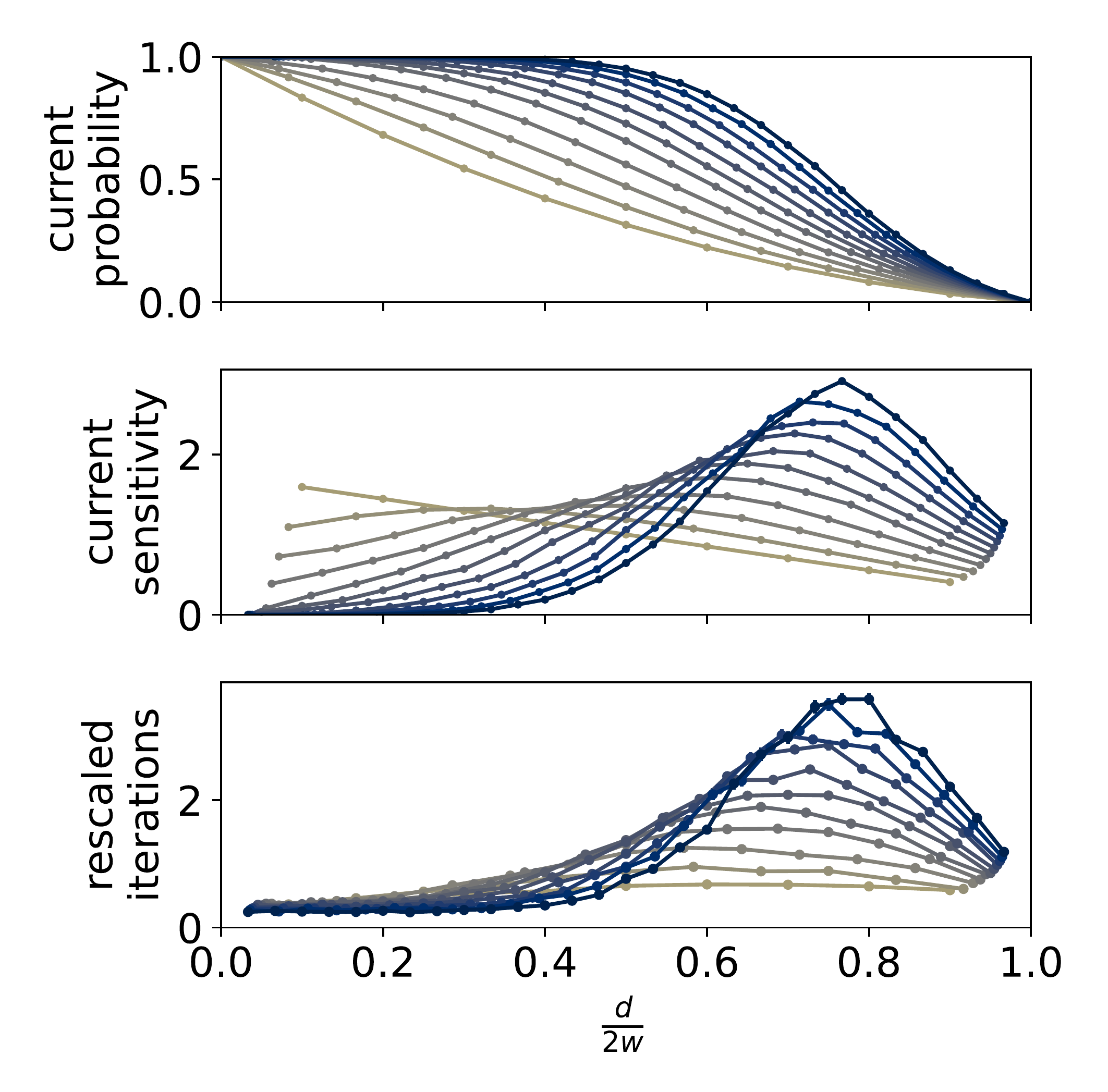}
    \caption{Current probability, current sensitivity, and rescaled iteration count for the star pruning method. The colors indicate the number of wires, with the darker curves corresponding to larger $w$.}
    \label{fig:wdIstar}
\end{figure}

Figure \ref{fig:wdIstar} plots the current probability, current sensitivity, and average number of iterations for the star pruning method. This time we must take $w$ from $15$ all the way down to $5$ before the network becomes disconnected. Here we see a significant discrepancy between the iteration count and the current sensitivity, particularly for the networks with the fewest wires and relatively few diodes. 

Evidently the pattern from section \ref{subsec:diodes} is not universal. It holds for well-connected networks with many paths from one node to another, such as the bipartite networks used in most of this paper, but it falters for more sparse networks, especially those with ``peninsular'' nodes that only have one wire connecting them to the rest of the network. Indeed, the $w=11$ curve in Figure \ref{fig:wdIrecursive}, the first in which we noticed a discrepancy between iteration count and current sensitivity, corresponds to the network with one node connected to the rest of the network by just a single wire. For the star method, the peninsulas appear at $w=8,7,6,5$---precisely the networks for which the discrepancy becomes most obvious in Figure \ref{fig:wdIstar}.

\section{Conclusion}
We have demonstrated a divide-and-concur iterative projection method for solving a percolation inverse problem in diode networks. Though the computational expense of our method grows rapidly with the size of the network, it grows much slower than that of the exhaustive approach, making our method far more practical for networks with more than a few nodes. 

We find that the projection algorithm generally requires more iterations to find a solution when one hides some of the current data, but only up to a point: If nearly all of the data are hidden the lack of constraints makes the problem very easy to solve. We have also used our algorithm to explore how the difficulty of the problem depends on the number of diodes in the network. Our results suggest that the most difficult networks are those for which the currents are most sensitive to the addition or removal of a single diode, although this pattern breaks down for highly pruned networks in which many nodes only have a single wire.

\section{Acknowledgments}
The author thanks Veit Elser (who helped to conceive of this problem), Paul McEuen, and James Sethna for valuable feedback and conversations.

\appendix
{ \section{Mapping to SAT}\label{app:SAT}
The SAT problem is typically formulated in conjunctive normal form (CNF): One is given a list of disjunctive clauses, such as $x_1 \,\vee\, \neg x_2 \,\vee\, x_3$, where each $x_i$ is a Boolean variable. Each $x_i$ can appear in several different clauses, and the task is to find values for the $x_i$ that make every clause true---that is, to satisfy the conjunction of all the clauses. 

It is possible to convert our task into a SAT problem. Let $W(i,j)$ be the set of (self-avoiding) paths from $i$ to $j$. Observing no current from $i$ to $j$ means every path is blocked, a fact which is easily expressed in CNF: 
\begin{equation}
    \bigwedge_{w\in W(i,j)} \left[\bigvee_{e\in w} x_e\right].
\end{equation}
If current \textit{is} observed, it means at least one path is unblocked, which is most naturally written as a disjunction of conjunctions:
\begin{equation}
    \bigvee_{w\in W(i,j)} \left[\bigwedge_{e\in w} \neg x_e\right].
    \label{eq:current}
\end{equation}
There is more than one way to convert a disjunction of conjunctions to a conjunction of disjunctions. The most direct method converts $c$ conjunctions of length $\ell$ to $\ell^c$ disjunctions. For example,
\begin{multline}
    (x_1 \,\wedge\, x_2 \,\wedge\, x_3) \,\vee\, (y_1 \,\wedge\, y_2 \,\wedge\, y_3)\\
    = (x_1 \,\vee\, y_1) \,\wedge\, (x_2 \,\vee\, y_1) \,\wedge\, (x_1 \,\vee\, y_2) \,\wedge\, ... \,\wedge\, (x_3 \,\vee\, y_3).
\end{multline}
However, this is likely to produce an unacceptable number of clauses. For our bipartite networks with $n=10$, there are $576$ self-voiding paths of length $9$ from one node to another (if one node is odd and the other is even) or $720$ paths of length $8$ (if both nodes are odd or both even). This would yield either $9^{576}$ or $8^{720}$ clauses, just for one current observation (without even considering paths of other lengths). A more practical method is to introduce extra Boolean variables: Let there be a $z_w$ for every $w\in W(i,j)$. Then
\begin{equation}
    \left[\bigvee_{w\in W(i,j)} z_w\right] \,\wedge\, \left[\bigwedge_{w\in W(i,j)} \left[\bigwedge_{e\in w}\neg z_w \,\vee\, \neg x_e\right]\right]
\end{equation}
has the same satisfiability as \eqref{eq:current}. Unfortunately, this still requires the introduction of a new variable for each path involved. 

It is worth noting that the CNF approach handles the no-current constraints quite naturally but is rather less elegant for the pairs that do have current, whereas our algorithm (Section \ref{sec:algo}) has a straightforward projection for the pairs with current but uses a quasi-projection to handle the no-current pairs. Perhaps a combined approach---using our method for the current pairs and CNF for the no-current pairs---could work more efficiently than either method alone.
}

\section{The quasi-projection to set $A$}\label{app:projA}
{Our quasi-projection to $A$ always finds a point in set $A$, but there are some circumstances in which it does not select the \textit{closest} point in $A$.

Suppose that no current was observed for pair $p$ and that, after the initial rounding of all $x_{ep}$ to $0$ or $1$ as prescribed in Eq.~\eqref{eq:Around}, two paths need to be blocked. That is, both paths had $x_{ep}<0.5$ on all of their edges $e$.} The extra distance for changing $x^A_{ep}$ from $0$ to $1$ is \begin{equation*}
    (1-x_{ep})^2-x_{ep}^2=1-2x_{ep}.
\end{equation*} For each path individually, the best move is to block the edge with the largest $x_{ep}$. Let $e_1$ be that edge for path $1$, and likewise for path $2$. For definiteness, suppose ${x_{e_1p}=x_{e_2p}=0.}1$, so that the added distance for blocking each path is $0.8$. But if the two paths share an edge $e_0\neq e_1,e_2$ with $x_{e_0p}=0$, then blocking $e_0$ is actually the best move overall, adding a distance of only $1$ compared to ${2\times0.8=1.6}$.

{Admittedly, this is a contrived example, and in practice such situations do not seem to prevent the algorithm from finding solutions.}

\section{Constructing a solution by hand}\label{app:trivialproof}
Suppose we know whether current can percolate from $i$ to $j$ for every $(i,j)$. For any $(i,j)$ that share a wire and for which current does not flow from $i\to j$, we must place a diode on wire $ij$ that will prevent current from $i\to j$. We have to place these diodes because otherwise current would flow from $i\to j$, in violation of our data. 

\begin{claim} This placement of diodes solves the problem.
\end{claim}

\begin{proof}
Suppose current can percolate from $i$ to $j$ in the actual circuit used to generate the data; that is, there exists a percolation path from $i$ to $j$. We can break the path into steps: $i\to k$, ..., $l\to j$. Then current also percolates from $i$ to $k$, so we must have left $i\to k$ unblocked in our solution circuit. The same holds for every other step in the path. Concatenating these steps together shows that current can percolate from $i$ to $j$ in our solution circuit.

In the converse direction, suppose current has a path from $i$ to $j$ in our solution circuit. Again, we can break the path into steps: $i\to k$, ..., $l\to j$. We would have blocked $i\to k$ in our solution circuit if the data had told us current does not flow from $i$ to $k$, so the fact that we did not implies that current can percolate from $i$ to $k$ in the actual circuit. The same holds for every other step in the path. Any complete dataset generated from an actual circuit has current transitivity: If we observe current from $a\to b$ and $b\to c$ then we must also observe current from $a\to c$. Thus, by the transitive property we conclude that current can percolate from $i$ to $j$ in the actual circuit.
\end{proof}

The reasoning in the converse direction breaks down if we do not have complete data. There may be percolations in our solution that are forbidden in the actual circuit.

However, if the network is completely connected then there is another construction that will work even with incomplete data: Place diodes blocking every directed edge $i\to j$ except those for which we know current can flow from $i$ to $j$. 

\begin{claim} This placement of diodes solves the problem: The solution may not have the same percolations as the circuit that generated the data, but it is consistent with all the available data.
\end{claim}

\begin{proof}
Suppose the data indicate that current can percolate from $i$ to $j$. Then we did not place a diode blocking $i\to j$ in our solution circuit, so there is certainly a path from $i$ to $j$ in our solution circuit. The contrapositive of this implication is that if there is no percolation from $i$ to $j$ in our circuit, then there is no percolation from $i$ to $j$ in the data table: The data table either says no percolation, or it contains no information on $(i,j)$. Either way, we are consistent with the data.

Now suppose there is a percolation path from $i$ to $j$ in our solution circuit, and break the path into steps: $i\to k$, ..., $l\to j$. We would not have left $i\to k$ unblocked in our solution circuit unless the data told us current flows from $i$ to $k$, so the data must indicate that current percolates from $i$ to $k$. The same holds for every other step in the path. By the transitive property discussed above, current can percolate from $i$ to $j$ in the circuit that generated the data. So the data table either says so, or it contains no information on $(i,j)$. Once again, we are consistent with the data.
\end{proof}

The hypothesis of complete connectivity is necessary: If the network is not completely connected, there may be a pair of nodes $(i,j)$ that do not share a wire, but do have percolation according to the data. The direct path $i\to j$ does not exist in this case, so we cannot guarantee that our solution allows percolation from $i$ to $j$.

There are surely other ways to construct solutions, but we consider these the simplest. In restricting ourselves to cases in which these constructions fail, we are in essence asking our algorithm to be more clever than we have been in making these constructions. At some point it becomes more time-consuming to come up with better constructions than to simply implement our algorithm.

\bibliography{diodes}

\providecommand{\noopsort}[1]{}\providecommand{\singleletter}[1]{#1}%
\begin{thebibliography}{30}%
\makeatletter
\providecommand \@ifxundefined [1]{%
 \@ifx{#1\undefined}
}%
\providecommand \@ifnum [1]{%
 \ifnum #1\expandafter \@firstoftwo
 \else \expandafter \@secondoftwo
 \fi
}%
\providecommand \@ifx [1]{%
 \ifx #1\expandafter \@firstoftwo
 \else \expandafter \@secondoftwo
 \fi
}%
\providecommand \natexlab [1]{#1}%
\providecommand \enquote  [1]{``#1''}%
\providecommand \bibnamefont  [1]{#1}%
\providecommand \bibfnamefont [1]{#1}%
\providecommand \citenamefont [1]{#1}%
\providecommand \href@noop [0]{\@secondoftwo}%
\providecommand \href [0]{\begingroup \@sanitize@url \@href}%
\providecommand \@href[1]{\@@startlink{#1}\@@href}%
\providecommand \@@href[1]{\endgroup#1\@@endlink}%
\providecommand \@sanitize@url [0]{\catcode `\\12\catcode `\$12\catcode
  `\&12\catcode `\#12\catcode `\^12\catcode `\_12\catcode `\%12\relax}%
\providecommand \@@startlink[1]{}%
\providecommand \@@endlink[0]{}%
\providecommand \url  [0]{\begingroup\@sanitize@url \@url }%
\providecommand \@url [1]{\endgroup\@href {#1}{\urlprefix }}%
\providecommand \urlprefix  [0]{URL }%
\providecommand \Eprint [0]{\href }%
\providecommand \doibase [0]{https://doi.org/}%
\providecommand \selectlanguage [0]{\@gobble}%
\providecommand \bibinfo  [0]{\@secondoftwo}%
\providecommand \bibfield  [0]{\@secondoftwo}%
\providecommand \translation [1]{[#1]}%
\providecommand \BibitemOpen [0]{}%
\providecommand \bibitemStop [0]{}%
\providecommand \bibitemNoStop [0]{.\EOS\space}%
\providecommand \EOS [0]{\spacefactor3000\relax}%
\providecommand \BibitemShut  [1]{\csname bibitem#1\endcsname}%
\let\auto@bib@innerbib\@empty
%</preamble>
\bibitem [{\citenamefont {Obukhov}(1980)}]{obukhov1980problem}%
  \BibitemOpen
  \bibfield  {author} {\bibinfo {author} {\bibfnamefont {S.}~\bibnamefont
  {Obukhov}},\ }\href
  {https://doi.org/https://doi.org/10.1016/0378-4371(80)90105-3} {\bibfield
  {journal} {\bibinfo  {journal} {Physica A: Statistical Mechanics and its
  Applications}\ }\textbf {\bibinfo {volume} {101}},\ \bibinfo {pages} {145}
  (\bibinfo {year} {1980})}\BibitemShut {NoStop}%
\bibitem [{\citenamefont {Takeuchi}\ \emph {et~al.}(2007)\citenamefont
  {Takeuchi}, \citenamefont {Kuroda}, \citenamefont {Chat{\'e}},\ and\
  \citenamefont {Sano}}]{takeuchi2007liquid}%
  \BibitemOpen
  \bibfield  {author} {\bibinfo {author} {\bibfnamefont {K.~A.}\ \bibnamefont
  {Takeuchi}}, \bibinfo {author} {\bibfnamefont {M.}~\bibnamefont {Kuroda}},
  \bibinfo {author} {\bibfnamefont {H.}~\bibnamefont {Chat{\'e}}},\ and\
  \bibinfo {author} {\bibfnamefont {M.}~\bibnamefont {Sano}},\ }\href@noop {}
  {\bibfield  {journal} {\bibinfo  {journal} {Physical review letters}\
  }\textbf {\bibinfo {volume} {99}},\ \bibinfo {pages} {234503} (\bibinfo
  {year} {2007})}\BibitemShut {NoStop}%
\bibitem [{\citenamefont {Tang}\ and\ \citenamefont
  {Leschhorn}(1992)}]{tang1992pinning}%
  \BibitemOpen
  \bibfield  {author} {\bibinfo {author} {\bibfnamefont {L.-H.}\ \bibnamefont
  {Tang}}\ and\ \bibinfo {author} {\bibfnamefont {H.}~\bibnamefont
  {Leschhorn}},\ }\href@noop {} {\bibfield  {journal} {\bibinfo  {journal}
  {Physical Review A}\ }\textbf {\bibinfo {volume} {45}},\ \bibinfo {pages}
  {R8309} (\bibinfo {year} {1992})}\BibitemShut {NoStop}%
\bibitem [{\citenamefont {Pomeau}(1986)}]{pomeau1986hydrodynamics}%
  \BibitemOpen
  \bibfield  {author} {\bibinfo {author} {\bibfnamefont {Y.}~\bibnamefont
  {Pomeau}},\ }\href@noop {} {\bibfield  {journal} {\bibinfo  {journal}
  {Physica D: Nonlinear Phenomena}\ }\textbf {\bibinfo {volume} {23}},\
  \bibinfo {pages} {3} (\bibinfo {year} {1986})}\BibitemShut {NoStop}%
\bibitem [{\citenamefont {Domany}\ and\ \citenamefont
  {Kinzel}(1984)}]{domany1984equivalence}%
  \BibitemOpen
  \bibfield  {author} {\bibinfo {author} {\bibfnamefont {E.}~\bibnamefont
  {Domany}}\ and\ \bibinfo {author} {\bibfnamefont {W.}~\bibnamefont
  {Kinzel}},\ }\href@noop {} {\bibfield  {journal} {\bibinfo  {journal}
  {Physical review letters}\ }\textbf {\bibinfo {volume} {53}},\ \bibinfo
  {pages} {311} (\bibinfo {year} {1984})}\BibitemShut {NoStop}%
\bibitem [{\citenamefont {Springer}\ and\ \citenamefont
  {Kenyon}(2021)}]{springer2021gameoflife}%
  \BibitemOpen
  \bibfield  {author} {\bibinfo {author} {\bibfnamefont {J.~M.}\ \bibnamefont
  {Springer}}\ and\ \bibinfo {author} {\bibfnamefont {G.~T.}\ \bibnamefont
  {Kenyon}},\ }in\ \href@noop {} {\emph {\bibinfo {booktitle} {2021
  International Joint Conference on Neural Networks (IJCNN)}}}\ (\bibinfo
  {organization} {IEEE},\ \bibinfo {year} {2021})\ pp.\ \bibinfo {pages}
  {1--8}\BibitemShut {NoStop}%
\bibitem [{\citenamefont {Elser}(2021)}]{elser2021reconstructing}%
  \BibitemOpen
  \bibfield  {author} {\bibinfo {author} {\bibfnamefont {V.}~\bibnamefont
  {Elser}},\ }\href {https://doi.org/10.1103/PhysRevE.104.034301} {\bibfield
  {journal} {\bibinfo  {journal} {Phys. Rev. E}\ }\textbf {\bibinfo {volume}
  {104}},\ \bibinfo {pages} {034301} (\bibinfo {year} {2021})}\BibitemShut
  {NoStop}%
\bibitem [{\citenamefont {Fan}\ \emph {et~al.}(2012)\citenamefont {Fan},
  \citenamefont {Liu}, \citenamefont {Li},\ and\ \citenamefont
  {Chen}}]{fan2012random}%
  \BibitemOpen
  \bibfield  {author} {\bibinfo {author} {\bibfnamefont {J.}~\bibnamefont
  {Fan}}, \bibinfo {author} {\bibfnamefont {M.}~\bibnamefont {Liu}}, \bibinfo
  {author} {\bibfnamefont {L.}~\bibnamefont {Li}},\ and\ \bibinfo {author}
  {\bibfnamefont {X.}~\bibnamefont {Chen}},\ }\href@noop {} {\bibfield
  {journal} {\bibinfo  {journal} {Physical Review E}\ }\textbf {\bibinfo
  {volume} {85}},\ \bibinfo {pages} {061110} (\bibinfo {year}
  {2012})}\BibitemShut {NoStop}%
\bibitem [{\citenamefont {Lee}\ \emph {et~al.}(2018)\citenamefont {Lee},
  \citenamefont {Kahng}, \citenamefont {Cho}, \citenamefont {Goh},\ and\
  \citenamefont {Lee}}]{lee2018complex}%
  \BibitemOpen
  \bibfield  {author} {\bibinfo {author} {\bibfnamefont {D.}~\bibnamefont
  {Lee}}, \bibinfo {author} {\bibfnamefont {B.}~\bibnamefont {Kahng}}, \bibinfo
  {author} {\bibfnamefont {Y.}~\bibnamefont {Cho}}, \bibinfo {author}
  {\bibfnamefont {K.-I.}\ \bibnamefont {Goh}},\ and\ \bibinfo {author}
  {\bibfnamefont {D.-S.}\ \bibnamefont {Lee}},\ }\href@noop {} {\bibfield
  {journal} {\bibinfo  {journal} {Journal of the Korean Physical Society}\
  }\textbf {\bibinfo {volume} {73}},\ \bibinfo {pages} {152} (\bibinfo {year}
  {2018})}\BibitemShut {NoStop}%
\bibitem [{\citenamefont {Colomer-de Sim{\'o}n}\ and\ \citenamefont
  {Bogun{\'a}}(2014)}]{colomer2014double}%
  \BibitemOpen
  \bibfield  {author} {\bibinfo {author} {\bibfnamefont {P.}~\bibnamefont
  {Colomer-de Sim{\'o}n}}\ and\ \bibinfo {author} {\bibfnamefont
  {M.}~\bibnamefont {Bogun{\'a}}},\ }\href@noop {} {\bibfield  {journal}
  {\bibinfo  {journal} {Physical Review X}\ }\textbf {\bibinfo {volume} {4}},\
  \bibinfo {pages} {041020} (\bibinfo {year} {2014})}\BibitemShut {NoStop}%
\bibitem [{\citenamefont {Miller}(2009)}]{miller2009cluster}%
  \BibitemOpen
  \bibfield  {author} {\bibinfo {author} {\bibfnamefont {J.~C.}\ \bibnamefont
  {Miller}},\ }\href@noop {} {\bibfield  {journal} {\bibinfo  {journal}
  {Physical Review E}\ }\textbf {\bibinfo {volume} {80}},\ \bibinfo {pages}
  {020901} (\bibinfo {year} {2009})}\BibitemShut {NoStop}%
\bibitem [{\citenamefont {Broadbent}\ and\ \citenamefont
  {Hammersly}(1957)}]{broadbent1957percolation}%
  \BibitemOpen
  \bibfield  {author} {\bibinfo {author} {\bibfnamefont {S.~R.}\ \bibnamefont
  {Broadbent}}\ and\ \bibinfo {author} {\bibfnamefont {J.~M.}\ \bibnamefont
  {Hammersly}},\ }\href@noop {} {\bibfield  {journal} {\bibinfo  {journal}
  {Math. Proc. Cambridge Philosophical Soc.}\ }\textbf {\bibinfo {volume}
  {53}},\ \bibinfo {pages} {629} (\bibinfo {year} {1957})}\BibitemShut
  {NoStop}%
\bibitem [{\citenamefont {Redner}(1982)}]{redner1982directed}%
  \BibitemOpen
  \bibfield  {author} {\bibinfo {author} {\bibfnamefont {S.}~\bibnamefont
  {Redner}},\ }\href@noop {} {\bibfield  {journal} {\bibinfo  {journal} {Phys.
  Rev. B}\ }\textbf {\bibinfo {volume} {25}},\ \bibinfo {pages} {3242}
  (\bibinfo {year} {1982})}\BibitemShut {NoStop}%
\bibitem [{\citenamefont {Redner}\ and\ \citenamefont
  {Mueller}(1982)}]{redner1982conductivity}%
  \BibitemOpen
  \bibfield  {author} {\bibinfo {author} {\bibfnamefont {S.}~\bibnamefont
  {Redner}}\ and\ \bibinfo {author} {\bibfnamefont {P.~R.}\ \bibnamefont
  {Mueller}},\ }\href@noop {} {\bibfield  {journal} {\bibinfo  {journal} {Phys.
  Rev. B}\ }\textbf {\bibinfo {volume} {26}},\ \bibinfo {pages} {5293}
  (\bibinfo {year} {1982})}\BibitemShut {NoStop}%
\bibitem [{\citenamefont {Verbavatz}\ and\ \citenamefont
  {Barthelemy}(2021)}]{verbavatz2021oneway}%
  \BibitemOpen
  \bibfield  {author} {\bibinfo {author} {\bibfnamefont {V.}~\bibnamefont
  {Verbavatz}}\ and\ \bibinfo {author} {\bibfnamefont {M.}~\bibnamefont
  {Barthelemy}},\ }\href@noop {} {\bibfield  {journal} {\bibinfo  {journal}
  {Physical Review E}\ }\textbf {\bibinfo {volume} {103}},\ \bibinfo {pages}
  {042313} (\bibinfo {year} {2021})}\BibitemShut {NoStop}%
\bibitem [{\citenamefont {De~Noronha}\ \emph {et~al.}(2018)\citenamefont
  {De~Noronha}, \citenamefont {Moreira}, \citenamefont {Vieira}, \citenamefont
  {Herrmann}, \citenamefont {Andrade~Jr},\ and\ \citenamefont
  {Carmona}}]{denoronha2018isotropic}%
  \BibitemOpen
  \bibfield  {author} {\bibinfo {author} {\bibfnamefont {A.~W.}\ \bibnamefont
  {De~Noronha}}, \bibinfo {author} {\bibfnamefont {A.~A.}\ \bibnamefont
  {Moreira}}, \bibinfo {author} {\bibfnamefont {A.~P.}\ \bibnamefont {Vieira}},
  \bibinfo {author} {\bibfnamefont {H.~J.}\ \bibnamefont {Herrmann}}, \bibinfo
  {author} {\bibfnamefont {J.~S.}\ \bibnamefont {Andrade~Jr}},\ and\ \bibinfo
  {author} {\bibfnamefont {H.~A.}\ \bibnamefont {Carmona}},\ }\href@noop {}
  {\bibfield  {journal} {\bibinfo  {journal} {Physical Review E}\ }\textbf
  {\bibinfo {volume} {98}},\ \bibinfo {pages} {062116} (\bibinfo {year}
  {2018})}\BibitemShut {NoStop}%
\bibitem [{\citenamefont {Runge}(2018)}]{runge2018causal}%
  \BibitemOpen
  \bibfield  {author} {\bibinfo {author} {\bibfnamefont {J.}~\bibnamefont
  {Runge}},\ }\href@noop {} {\bibfield  {journal} {\bibinfo  {journal} {Chaos:
  An Interdisciplinary Journal of Nonlinear Science}\ }\textbf {\bibinfo
  {volume} {28}},\ \bibinfo {pages} {075310} (\bibinfo {year}
  {2018})}\BibitemShut {NoStop}%
\bibitem [{\citenamefont {Sontag}(2008)}]{sontag2008steady}%
  \BibitemOpen
  \bibfield  {author} {\bibinfo {author} {\bibfnamefont {E.}~\bibnamefont
  {Sontag}},\ }\href@noop {} {\bibfield  {journal} {\bibinfo  {journal} {Essays
  in Biochemistry}\ }\textbf {\bibinfo {volume} {45}},\ \bibinfo {pages} {161}
  (\bibinfo {year} {2008})}\BibitemShut {NoStop}%
\bibitem [{\citenamefont {Angulo}\ \emph {et~al.}(2017)\citenamefont {Angulo},
  \citenamefont {Moreno}, \citenamefont {Lippner}, \citenamefont
  {Barab{\'a}si},\ and\ \citenamefont {Liu}}]{angulo2017fundamental}%
  \BibitemOpen
  \bibfield  {author} {\bibinfo {author} {\bibfnamefont {M.~T.}\ \bibnamefont
  {Angulo}}, \bibinfo {author} {\bibfnamefont {J.~A.}\ \bibnamefont {Moreno}},
  \bibinfo {author} {\bibfnamefont {G.}~\bibnamefont {Lippner}}, \bibinfo
  {author} {\bibfnamefont {A.-L.}\ \bibnamefont {Barab{\'a}si}},\ and\ \bibinfo
  {author} {\bibfnamefont {Y.-Y.}\ \bibnamefont {Liu}},\ }\href@noop {}
  {\bibfield  {journal} {\bibinfo  {journal} {Journal of the Royal Society
  Interface}\ }\textbf {\bibinfo {volume} {14}},\ \bibinfo {pages} {20160966}
  (\bibinfo {year} {2017})}\BibitemShut {NoStop}%
\bibitem [{\citenamefont {Peixoto}(2019)}]{peixoto2019community}%
  \BibitemOpen
  \bibfield  {author} {\bibinfo {author} {\bibfnamefont {T.~P.}\ \bibnamefont
  {Peixoto}},\ }\href@noop {} {\bibfield  {journal} {\bibinfo  {journal}
  {Physical review letters}\ }\textbf {\bibinfo {volume} {123}},\ \bibinfo
  {pages} {128301} (\bibinfo {year} {2019})}\BibitemShut {NoStop}%
\bibitem [{\citenamefont {Cook}(1971)}]{cook1971complexity}%
  \BibitemOpen
  \bibfield  {author} {\bibinfo {author} {\bibfnamefont {S.~A.}\ \bibnamefont
  {Cook}},\ }in\ \href@noop {} {\emph {\bibinfo {booktitle} {Proceedings of the
  third annual ACM symposium on Theory of computing}}}\ (\bibinfo {year}
  {1971})\ pp.\ \bibinfo {pages} {151--158}\BibitemShut {NoStop}%
\bibitem [{\citenamefont {Lynce}\ and\ \citenamefont
  {Ouaknine}(2006)}]{lynce2006sudoku}%
  \BibitemOpen
  \bibfield  {author} {\bibinfo {author} {\bibfnamefont {I.}~\bibnamefont
  {Lynce}}\ and\ \bibinfo {author} {\bibfnamefont {J.}~\bibnamefont
  {Ouaknine}},\ }\href@noop {} {\bibfield  {journal} {\bibinfo  {journal}
  {ISAIM}\ }\textbf {\bibinfo {volume} {11}},\ \bibinfo {pages} {6} (\bibinfo
  {year} {2006})}\BibitemShut {NoStop}%
\bibitem [{\citenamefont {Massacci}\ and\ \citenamefont
  {Marraro}(2000)}]{massacci2000cryptanalysis}%
  \BibitemOpen
  \bibfield  {author} {\bibinfo {author} {\bibfnamefont {F.}~\bibnamefont
  {Massacci}}\ and\ \bibinfo {author} {\bibfnamefont {L.}~\bibnamefont
  {Marraro}},\ }\href@noop {} {\bibfield  {journal} {\bibinfo  {journal}
  {Journal of Automated Reasoning}\ }\textbf {\bibinfo {volume} {24}},\
  \bibinfo {pages} {165} (\bibinfo {year} {2000})}\BibitemShut {NoStop}%
\bibitem [{\citenamefont {Bright}\ \emph {et~al.}(2019)\citenamefont {Bright},
  \citenamefont {Gerhard}, \citenamefont {Kotsireas},\ and\ \citenamefont
  {Ganesh}}]{bright2019effectiveSAT}%
  \BibitemOpen
  \bibfield  {author} {\bibinfo {author} {\bibfnamefont {C.}~\bibnamefont
  {Bright}}, \bibinfo {author} {\bibfnamefont {J.}~\bibnamefont {Gerhard}},
  \bibinfo {author} {\bibfnamefont {I.}~\bibnamefont {Kotsireas}},\ and\
  \bibinfo {author} {\bibfnamefont {V.}~\bibnamefont {Ganesh}},\ }in\
  \href@noop {} {\emph {\bibinfo {booktitle} {Maple Conference}}}\ (\bibinfo
  {organization} {Springer},\ \bibinfo {year} {2019})\ pp.\ \bibinfo {pages}
  {205--219}\BibitemShut {NoStop}%
\bibitem [{\citenamefont {Molnár}\ \emph {et~al.}(2018)\citenamefont
  {Molnár}, \citenamefont {Varga}, \citenamefont {Toroczkai},\ and\
  \citenamefont {Ercsey-Ravasz}}]{molnar2018maxsat}%
  \BibitemOpen
  \bibfield  {author} {\bibinfo {author} {\bibfnamefont {B.}~\bibnamefont
  {Molnár}}, \bibinfo {author} {\bibfnamefont {M.}~\bibnamefont {Varga}},
  \bibinfo {author} {\bibfnamefont {Z.}~\bibnamefont {Toroczkai}},\ and\
  \bibinfo {author} {\bibfnamefont {M.}~\bibnamefont {Ercsey-Ravasz}},\ }\href
  {https://doi.org/10.48550/ARXIV.1801.06620} {\bibinfo {title} {A
  high-performance analog max-sat solver and its application to ramsey
  numbers}} (\bibinfo {year} {2018})\BibitemShut {NoStop}%
\bibitem [{\citenamefont {Ercsey-Ravasz}\ and\ \citenamefont
  {Toroczkai}(2011)}]{ercsey2011sat}%
  \BibitemOpen
  \bibfield  {author} {\bibinfo {author} {\bibfnamefont {M.}~\bibnamefont
  {Ercsey-Ravasz}}\ and\ \bibinfo {author} {\bibfnamefont {Z.}~\bibnamefont
  {Toroczkai}},\ }\href@noop {} {\bibfield  {journal} {\bibinfo  {journal}
  {Nature Physics}\ }\textbf {\bibinfo {volume} {7}},\ \bibinfo {pages} {966}
  (\bibinfo {year} {2011})}\BibitemShut {NoStop}%
\bibitem [{\citenamefont {Sorensson}\ and\ \citenamefont
  {Een}(2005)}]{sorensson2005minisat}%
  \BibitemOpen
  \bibfield  {author} {\bibinfo {author} {\bibfnamefont {N.}~\bibnamefont
  {Sorensson}}\ and\ \bibinfo {author} {\bibfnamefont {N.}~\bibnamefont
  {Een}},\ }\href@noop {} {\bibfield  {journal} {\bibinfo  {journal} {SAT}\
  }\textbf {\bibinfo {volume} {2005}},\ \bibinfo {pages} {1} (\bibinfo {year}
  {2005})}\BibitemShut {NoStop}%
\bibitem [{\citenamefont {Biere}(2013)}]{biere2013lingeling}%
  \BibitemOpen
  \bibfield  {author} {\bibinfo {author} {\bibfnamefont {A.}~\bibnamefont
  {Biere}},\ }\href@noop {} {\bibfield  {journal} {\bibinfo  {journal}
  {Proceedings of SAT competition}\ }\textbf {\bibinfo {volume} {2013}},\
  \bibinfo {pages} {1} (\bibinfo {year} {2013})}\BibitemShut {NoStop}%
\bibitem [{\citenamefont {Gravel}\ and\ \citenamefont
  {Elser}(2008)}]{gravel2008divide}%
  \BibitemOpen
  \bibfield  {author} {\bibinfo {author} {\bibfnamefont {S.}~\bibnamefont
  {Gravel}}\ and\ \bibinfo {author} {\bibfnamefont {V.}~\bibnamefont {Elser}},\
  }\href@noop {} {\bibfield  {journal} {\bibinfo  {journal} {Physical Review
  E}\ }\textbf {\bibinfo {volume} {78}},\ \bibinfo {pages} {036706} (\bibinfo
  {year} {2008})}\BibitemShut {NoStop}%
\bibitem [{\citenamefont {Arag{\'o}n~Artacho}\ \emph
  {et~al.}(2020)\citenamefont {Arag{\'o}n~Artacho}, \citenamefont {Campoy},\
  and\ \citenamefont {Tam}}]{aragon2020douglas}%
  \BibitemOpen
  \bibfield  {author} {\bibinfo {author} {\bibfnamefont {F.~J.}\ \bibnamefont
  {Arag{\'o}n~Artacho}}, \bibinfo {author} {\bibfnamefont {R.}~\bibnamefont
  {Campoy}},\ and\ \bibinfo {author} {\bibfnamefont {M.~K.}\ \bibnamefont
  {Tam}},\ }\href@noop {} {\bibfield  {journal} {\bibinfo  {journal}
  {Mathematical Methods of Operations Research}\ }\textbf {\bibinfo {volume}
  {91}},\ \bibinfo {pages} {201} (\bibinfo {year} {2020})}\BibitemShut
  {NoStop}%
\end{thebibliography}%

\end{document}